\documentclass[conference]{IEEEtran}
\usepackage{cite}
\usepackage{amsmath,amssymb,amsfonts}
\usepackage{algorithmic}

\usepackage{graphicx}
\usepackage{textcomp}
\usepackage{xcolor}
\usepackage{array}
\usepackage{amsthm}
\usepackage{caption}
\usepackage{hyperref}
\newtheorem{theorem}{Theorem}
\newtheorem{lemma}{Lemma}
\newtheorem{corollary}{Corollary}
\usepackage{subcaption}

\ifCLASSINFOpdf
\else
\fi
\begin{document}
%
\title{Integrated Wake-Up Radio and MIMO Solution for Cellular IoT Networks}

\author{
    \IEEEauthorblockN{Israa Khaled\IEEEauthorrefmark{1}, Ammar El Falou\IEEEauthorrefmark{2}, Nour Kouzayha\IEEEauthorrefmark{2}, Charlotte Langlais\IEEEauthorrefmark{3}}
    \IEEEauthorblockA{\IEEEauthorrefmark{1}Laboratoire de Traitement et Communication de l’Information (LTCI), Télécom Paris, Institut Polytechnique de Paris, France}
    \IEEEauthorblockA{\IEEEauthorrefmark{2} CEMSE Division, King Abdullah University of Science and Technology (KAUST), Saudi Arabia}
    \IEEEauthorblockA{\IEEEauthorrefmark{3} Mathematical and Electrical Department, CNRS UMR 6285 Lab-STICC, IMT Atlantique, Brest, France}
}


%


\maketitle

\begin{abstract}

    Wake-up radio (WUR) is a technology designed to enhance the energy efficiency of Internet of Things (IoT) networks and extend device battery life. While most studies focus on WUR performance with single-antenna base stations, this paper investigates the multiple-input multiple-output (MIMO) technology to improve device energy saving and extend the coverage of wake-up signals. By leveraging MIMO beamforming, the transmitted energy can be spatially focused toward the intended IoT devices, with high beamforming gain and minimal inter-device interference. We develop a preliminary analytical framework using stochastic geometry to evaluate the wake-up success probability of WUR-MIMO in multi-cell cellular IoT networks, when the number of antennas equals $2 \times (\text{number of devices}) - 1$. Monte Carlo simulations show that, relative to a single-antenna WUR baseline, MIMO beamforming significantly enhances wake-up reliability when this antenna configuration is applied, mitigates more than 50\% of false activations across all settings, and thereby prolongs the lifetime of IoT devices.
    
\end{abstract}


%
\IEEEpeerreviewmaketitle

\begin{IEEEkeywords}
Wake-up radio, multiple-input multiple-output, beamforming, Internet of Things, stochastic geometry.
\end{IEEEkeywords}
\section{Introduction}
The Internet of Things (IoT) will play a central role in the sixth generation (6G) of wireless networks by enabling ultra-reliable, low-latency communication and massive connectivity, supporting applications such as smart cities, autonomous transport, smart agriculture, and healthcare. However, large-scale IoT deployment faces significant energy constraints, as most devices rely on limited-capacity batteries that are difficult to replace or recharge, especially in remote areas~\cite{piyare2017ultra}.

Device energy-saving techniques for large-scale IoT networks have gained significant attention. \textit{Energy harvesting} enables devices to capture ambient energy (e.g., solar, wind, radio-frequency (RF) signals), reducing battery dependence but suffering from variability in environmental conditions~\cite{ng2012energy}. \textit{Resource allocation} strategies minimize power consumption by optimizing communication and computation, though they often incur high computational and signaling overhead in large-scale systems~\cite{zhang2015energy}. \textit{Discontinuous reception (DRX)} reduces energy use by periodically turning off receivers, but misalignment with traffic can increase latency and delay urgent transmissions~\cite{jha2014device}.

Another promising approach is the \textit{wake-up radio (WUR)}, standardized in IEEE 802.11ba and included in 3GPP Release~18~\cite{TR18-3GPP,hoglund20243gpp}. WURs provide a hardware solution with listening power consumption orders of magnitude lower than that of low-power radios, addressing idle listening, overhearing, continuous transmissions, and data latency. Devices remain in deep sleep until an external RF signal exceeds a threshold; they then perform their tasks and return to sleep, minimizing energy drain and extending battery life~\cite{mercier_ultra-low_2015,mercier2022low,shellhammer2023ieee}. According to \cite{hoglund20243gpp}, significant device power savings can be achieved with the WUR technique compared to existing DRX/extended DRX.

Most research on WUR has centered on single-antenna transmitters which emit power omni-directionally without considering the specific locations of IoT devices~\cite{kouzayha_joint_2018}. This single-antenna WUR approach, however, suffers from false alarms, as the IoT devices may be erroneously activated by unintended signals~\cite{Khalifa2024Energy}. Moreover, the WUR range is constrained by the limited sensitivity of the added wake-up receiver hardware~\cite{benbuk_leveraging_2020,mercier2022low}. Enhancing sensitivity generally increases power consumption as it requires active circuit components at the device's side. Alternatively, digital baseband processing offers a pathway to balance sensitivity and energy efficiency effectively. For instance, a minimum energy channel coding scheme improves sensitivity while maintaining low energy consumption~\cite{djidi_mees-WUR_2021}.

To extend the coverage of the wake-up signal (WUS), we leverage the spatial diversity of multiple-input multiple-output (MIMO) technology by equipping the base station (BS) with multiple antennas~\cite{falou_precoded_2022}. 
Assuming full channel state information (CSI) knowledge, we employ MIMO beamforming to generate a precoded WUS that activates the targeted IoT devices~\cite{falou_precoded_2022}. We analyze how different MIMO beamforming strategies influence the successful wake-up of IoT devices within a geometric channel model in a multi-cell scenario. Additionally, in \cite{tannous2025mimo}, we evaluate MIMO beamforming performance in a single-cell IoT network under Rayleigh fading, considering both successful and false wake-up probabilities. The promising results from integrating WUR and MIMO technologies—denoted here as WUR-MIMO—motivate this work to leverage stochastic geometry, which provides a realistic and tractable framework for modeling MIMO systems in large-scale wireless networks~\cite{andrews2011tractable}. Stochastic geometry has also been applied to evaluate single-antenna WUR solutions~\cite{kouzayha_joint_2018} and RF energy harvesting~\cite{sakr2015analysis}. We focus on two main metrics, namely (i) \textit{the success wake-up probability}, defined as the probability of waking up the IoT device of interest, and (ii) \textit{the false wake-up probability}, characterizing the probability of false alarms. 


\textit{Notation:} 
\textbf{A}, \textbf{a} and  $a$ denote matrix, vector and scalar, respectively.  $ (\cdot) ^H $  represents the Hermitian transpose. $\Gamma(\alpha, \beta)$  is a gamma random variable with shape parameter $\alpha$ and scale parameter $\beta$. The Euclidean norm is denoted by $\lVert . \rVert$, and the complementary incomplete Beta function is given by $B'(q,l,z)=\int_z^1t^{q-1}(1-t)^{l-1}dt$.

\section{System Model}\label{sec:system}

We consider a network model where BSs are spatially distributed following a Poisson Point Process (PPP), denoted as $\Phi=\{\textbf{y}_b\}, b=1, 2, \cdots\in \mathbb{R}^2$ with intensity $\lambda$~ $[\text{BSs}/m^2]$ in the Euclidean plane. Each BS has $N_t$ antennas and serves $M_w$ single-antenna IoT devices or wireless sensors (WSs) within its coverage area as illustrated in Fig. \ref{fig:system_model}. Each WS is associated with its nearest BS, resulting in the division of the geographical area into Voronoi cells. The $w$-th WS connected to the $b$-th BS is denoted by WS$_{w,b}$. The channel vector between the $b'$-th BS and WS$_{w,b}$, represented as $\textbf{h}_{w,b}^{b'}$, follows a Circularly Symmetric Complex Gaussian distribution with zero mean and identity covariance matrix, i.e., $\textbf{h}_{w,b}^{b'}\sim \mathcal{CN}(0,\textbf{I}_{N_t})$. Denote by $\textbf{x}_{w,b}$ the location of WS$_{w,b}$ in $\mathbb{R}^2$. The total channel vector between the $b'$-th BS and WS$_{w,b}$ is given by $ \textbf{h}_{w,b}^{b'}(r_{w,b}^{b'})^{-\alpha/2}$, with $\alpha>2$ the pathloss exponent and $r_{w,b}^{b'}=\lVert \textbf{y}_{b'} -\textbf{x}_{w,b} \rVert$ the euclidean distance separating the BS at $\textbf{y}_{b'}$ from WS$_{w,b}$. Let $a\le1$ denote the RF-DC conversion efficiency characterizing the efficiency of the added wake-up receiver.

Let $\textbf{F}_b=\left[\textbf{f}_{1,b}, \cdots, \textbf{f}_{M_w,b}\right]$ represents the total beamforming matrix of all WSs connected to the $b$-th BS, where $\textbf{f}_{w,b}$ is the beamforming vector used by the BS to generate the precoded WUS intended for wireless sensor  WS$_{w,b}$. In this paper, we employ the ZF beamforming technique\footnote{ZF beamforming is adopted for its effective interference suppression; however, it incurs high computational complexity and channel estimation overhead in large-scale IoT networks. Angle-domain beamforming \cite{Israakhaled2020WSA} offers a lower-complexity alternative with reduced overhead and is left for future work. Moreover, \cite{khaled2021multi} could be used to represent angle-domain inter-device interference.} to effectively eliminate inter-device interference. The resulting beamforming matrix, $\textbf{F}_b$, is given by
\begin{equation}
\textbf{F}_{b}=\frac{{\textbf{H}_b}^H}{{\textbf{H}_b}{\textbf{H}_b}^H},
\end{equation}
where $\textbf{H}_b=\left[\textbf{h}_{1,b}^b, \cdots, \textbf{h}_{M_w,b}^b\right]$ is the channel matrix between the  $b$-th BS and its $M_w$ connected devices. Without loss of generality, we focus on the typical device WS$_{1,o}$ positioned at the origin and connected to BS $o$, i.e., {$r_{1,o}^b=\lVert \textbf{y}_b \rVert$}~\cite{chiu_stochastic_2013}. The total power $ {P}_{r}$ received at the typical device is given in~(\ref{eq:Pr}) at the top of the next page, 
	\begin{table*}[t]
    \normalsize
	\begin{equation}\label{eq:Pr}
       \resizebox{\textwidth}{!}{${P_r}=a p_{1,o} \eta_{1,o} \| \textbf{h}_{1,o}^{o} \textbf{f}_{1,o}\|^2 \| \textbf{y}_{o}\|^{-\alpha} +\\ a \sum_{n=2}^{M_w}p_{n,o}\eta_{n,o}\|\textbf{h}_{1,o}^{o}\textbf{f}_{n,o}\|^2\|\textbf{y}_{o}\|^{-\alpha}+a\sum_{\textbf{y}_b\in \Phi^o}\sum_{m=1}^{M_w}p_{m,b}\eta_{m,b}\|\textbf{h}_{1,o}^b\textbf{f}_{m,b}\|^2\|\textbf{y}_b\|^{-\alpha}.$}
	\end{equation}
	\hrulefill
	\end{table*}
where $\Phi^o=\Phi  \setminus \textbf{y}_{o}$ is the set of interfering BSs, $p_{w,b}$ is the power allocated to the WS$_{w,b}$, and  $\eta_{w,b}$ is the partial beamforming normalization factor at WS$_{w,b}$ given by 
\begin{equation}\label{eq:eta}
\eta_{w,b}=
\frac{1}{ \|\textbf{f}_{w,b}\|^2}.
\end{equation}
The first term in (\ref{eq:Pr}) corresponds to the desired power received at the typical device, while the sum of the second and third terms represents the total interference power denoted as $I$ accounting for intra- and inter-cell interference.
The power allocation (PA) matrix, $\textbf{P}_b=[p_{1,b}, \cdots , p_{M_w,b}]$ of the $M_w$ IoT devices connected to the $b$-th cell, satisfies the total power constraint, i.e., $\|\textbf{P}_b \|=P_t$. We assume a uniform PA among these devices, i.e., $p_{w,b}=\frac{P_t}{M_w} \ (\forall  w \in \{1, \cdots, M_w\}$ and $\textbf{y}_b \in \Phi$), to meet the total transmitted power constraint.  

\begin{figure}[!t]
\centering
\includegraphics[scale=0.3]{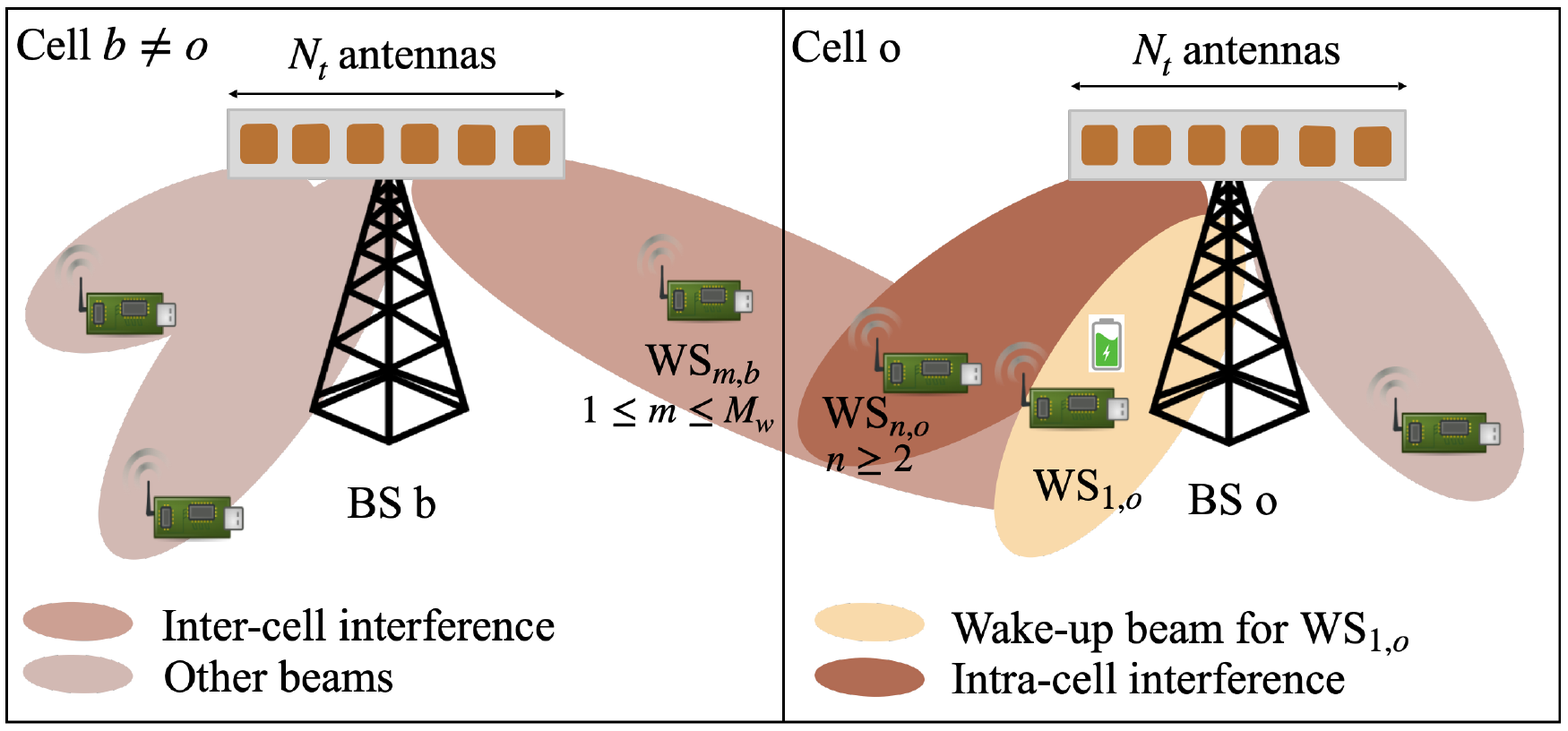}
    \vspace{-1cm}
   \caption{WUR-MIMO system model.}  \label{fig:system_model}
\end{figure}
\section{Wake-up Probability Analysis}\label{sec:derivation}
In this section, we develop an analytical framework to derive the probability of successful wake-up, defined as the event where the BS successfully wakes the typical device via precoded WUS transmission. Due to the complexity of the general case, a tractable expression is obtained for the specific setting $N_t = 2M_w - 1$. To further assess system performance, we also introduce the false wake-up probability, and evaluate both metrics via Monte-Carlo simulations in Section \ref{sec:simu}.

\subsection{Success Wake-up Probability}
A successful wake-up occurs when the BS sends a WUS to a device, and the received signal power $P_r$ exceeds a predefined threshold, $T$.  Assuming that the BS distributes its transmit power among all $M_w$ devices, the success wake-up probability, denoted as $ P_{\text{wu}}^s$ is defined as follows 
\begin{equation}
 P_{\text{wu}}^s \triangleq \mathbb{P}({P_r}\ge T ).
\label{eq:Pwu}
\end{equation}

The success wake-up probability, $P_{\text{wu}}^s$, corresponds to the complementary cumulative distribution function (CCDF) of the received power ${P_r}$. To compute $P_{\text{wu}}^s$, we derive the Laplace transform $\mathcal{L}_{{P_r}}(s)$ of ${P_r}$ and apply its inverse to obtain the CCDF, following the approach in~\cite{kouzayha_joint_2018}. Thus, $P_{\text{wu}}^s$ is given as
\begin{equation}
 P_{\text{wu}}^s =1-\mathcal{L}^{-1}\left\{\frac{1}{s} \mathbb{E}\left[ e^{-s{P_r}}\right ]\right\}.
\label{eq:Pwu}
\end{equation}

The Laplace transform $\mathcal{L}_{{P_r}}(s)$ is derived in Lemma 1 as: 
\begin{lemma}
In an IoT network, where the BSs are distributed according to a PPP with density $\lambda$, and each BS transmits with power \( P_t \) and is equipped with \( N_t = 2M_w - 1 \) antennas to serve \( M_w \) WSs per cell, the Laplace transform of the total received power at the typical device $\mathcal{L}_{{P_r}}(s)$ is given as
\begin{equation}
   \mathcal{L}_{{P_r}}(s)=e^{-\lambda \left[\frac{saP_t }{M_w}\right]^{\frac{2}{\alpha}} \Delta},
\end{equation}
where $\Delta=\frac{2\pi}{\alpha}\sum_{m=1}^{M_w} \binom {M_w}m B'(q,l,0)$, $q=\frac{2}{\alpha}-m+M_w$, and $l=-\frac{2}{\alpha} + m$.
    \begin{IEEEproof} The Laplace transform of the received power given in ~(\ref{eq:Pr}) can be derived as follows
         \begin{equation} \label{eq1}
\begin{aligned}
&\mathcal{L}_{{P_r}}(s)\\
&\overset{(a)}{=}\mathbb{E}_{\Phi}\Biggl\{  \prod_{\textbf{y}_b\in\Phi} \mathbb{E}_{\textbf{h}_{1,{o}}^b} \left[ e^{-s a \sum_{m=1}^{M_w} p_{m,b}\eta_{m,b}\|\textbf{h}_{1,o}^b\textbf{f}_{m,b}\|^2\|\textbf{y}_b\|^{-\alpha}}\right] \Biggl\} \\
& \overset{(b)}{=}\mathbb{E}_{\Phi}\Biggl\{ \mathbb{E}_{\textbf{h}_{1,o}^{o}} \left[ e^{-s a\sum_{n=1}^{M_w} p_{n,o}\eta_{n,o}\|\textbf{h}_{1,o}^{o}\textbf{f}_{n,o}\|^2\|\textbf{y}_{o}\|^{-\alpha}}\right] \\ 
&\times
\prod_{\textbf{y}_b\in\Phi^o } \mathbb{E}_{\textbf{h}_{1,{o}}^b} \left[ e^{-s a \sum_{m=1}^{M_w} p_{m,b}\eta_{m,b} \|\textbf{h}_{1,o}^b\textbf{f}_{m,b}\|^2\|\textbf{y}_b\|^{-\alpha}}\right]\Biggl\} \\
&\resizebox{1\hsize}{!}{$ \overset{(c)}{=}\mathbb{E}_{\Phi}\Biggl\{ \mathbb{E}_{\mu} \left[ e^{-s a \frac{P_t}{M_w}\mu\| \textbf{y}_{o}\|^{-\alpha}}\right]  \prod_{\textbf{y}_b\in\Phi^o} \mathbb{E}_{\nu_b} \left[ e^{-s a \frac{P_t}{M_w} \nu_b\|\textbf{y}_b\|^{-\alpha}}\right] \Biggl\}$},
\end{aligned}
\end{equation} 
where (a) follows from the independent and identically distributed property (i.i.d) of $\textbf{h}_{1,o}^b$ and its further independence from the PPP $\Phi$, (b) follows from the i.i.d property of $\textbf{h}_{1,o}^b$, and (c) follows from substituting the values of $p_{w,b}$ and $\eta_{w,b}$ and defining $\mu \triangleq \sum_{n=1}^{M_w}\frac{\|\textbf{h}_{1,o}^{o}\textbf{f}_{n,o}\|^2}{\|\textbf{f}_{n,o}\|^2}$  and $\nu_b \underset{b\neq o}{\triangleq} \sum_{m=1}^{M_w}\frac{\|\textbf{h}_{1,o}^{b}\textbf{f}_{m,b}\|^2}{\|\textbf{f}_{m,b}\|^2} $. According to~\cite{afify2016unified}, we can approximate $\mu \sim \Gamma(N_t-M_w+1,1)$ and $\nu_b \sim \Gamma (M_w,1)$. Thus, $\mathcal{L}_{{P_r}}(s)$ can be rewritten as follows 
 \begin{equation} \label{eq1}
\begin{split}
&\mathcal{L}_{{P_r}}(s) \overset{(d)}{=}\mathbb{E}_{\Phi}\Biggl\{ \frac{1}{(1+\frac{saP_t}{M_w} \|\textbf{y}_{o}\|^{-\alpha}) ^{N_t-M_w+1}}  \\ 
& \qquad \prod_{\textbf{y}_b\in\Phi^o} \frac{1}{(1+\frac{saP_t}{M_w}\|\textbf{y}_{b}\|^{-\alpha}) ^{M_w}} \Biggl\} \\
& \overset{(e)}{=}\mathbb{E}_{\Phi}\Biggl\{  g(\textbf{y}_o)^{N_t-M_w+1} \prod_{\textbf{y}_b\in\Phi^o} g(\textbf{y}_b) ^{M_w}  \Biggl\}, \\
& \overset{(f)}{=} \exp \left(-2\pi\lambda \int_{0}^{\infty}\left(1-\frac{1}{(1+\frac{saP_t}{M_w}v^{-\alpha})^{M_w}} \right)vdv\right),  \\
&\overset{(g)}{=}  \resizebox{0.9\hsize}{!}{$\exp \left(-2\pi\lambda \sum_{m=1}^{M_w} \binom {M_w}m \frac{1}{{\alpha}}\left[\frac{saP_t}{M_w}\right]^{\frac{2}{\alpha}}\int_{0}^{1}u^{q-1}(1-u)^{l-1}  du\right)$},
\end{split}
\end{equation}
where (d) follows from the Laplace transform of $\mu$ and $\nu_b$ following the Gamma distribution, (e) follows from setting $g(\textbf{y}_b)=\frac{1}{(1+\frac{saP_t}{M_w}\|\textbf{y}_{b}\|^{-\alpha})}$, with $\textbf{y}_b \in \Phi$, (f) follows from assuming $N_t=2M_w-1$\footnote{The constraint $N_t = 2M_w - 1$ is required to ensure closed-form tractability of the PGFL-based derivation. Addressing the integral complexity without this constraint, e.g., via a dominant BS approximation as done in  \cite{kouzayha2017analysis}, is left for future work.} and from applying the PGFL of the PPP. Finally, $q=\frac{2}{\alpha}-m+M_w$ and $l=-\frac{2}{\alpha} + m$ in (g) which is obtained as in~\cite{arsal_coverage_2021} via a change of variables, specifically $u = \frac{1}{1+ \frac{saP_t}{M_w}v^{-\alpha}}$.  We apply the following  manipulations
\begin{equation}\label{}
		\begin{cases}
			v=\left[\frac{usaP_t}{M_w(1-u)}\right]^{\frac{1}{\alpha}}, & \\ 
			dv=\frac{1}{{\alpha}}\left[\frac{saP_t}{M_w}\right]^{\frac{1}{\alpha}}u^{\frac{1}{\alpha}-1}(1-u)^{-\frac{1}{\alpha}-1}du. 
		\end{cases}
	\end{equation}
    \end{IEEEproof}
\end{lemma}
Finally, $P_{\text{wu}}^s$ is given in the following theorem.
\begin{theorem}
In an IoT network, where the BSs are distributed according to a PPP with density \( \lambda \), and each BS transmits with power \( P_t \) and is equipped with \( N_t = 2M_w - 1 \) antennas to serve \( M_w \) WSs per cell, the success wake-up probability \( P_{\text{wu}}^s \) for a generic device, considering an RF-DC conversion factor, \( a \), and a received power threshold, \( T \), is given by
\begin{multline}\label{eq:Pwu_exact}
     P_{\text{wu}}^s = \int_{0}^{\infty}\frac{e^{-xT}}{\pi x}e^{-\lambda  \left(\frac{aP_t}{M_w}\right)^{2/\alpha}\cos(2\pi/\alpha)x^{2/\alpha}\Delta} \\
    \sin\left(\lambda \left(\frac{aP_t}{M_w}\right)^{2/\alpha}\sin(2\pi/\alpha)x^{2/\alpha}\Delta\right) dx.
\end{multline}
\label{theorem:Ps}
\end{theorem}

\begin{proof}
See Appendix.
\end{proof}
Note that \eqref{eq:Pwu_exact} cannot be obtained in closed form for a general pathloss exponent; however, it can be efficiently evaluated numerically. A closed-form expression for \eqref{eq:Pwu_exact} is available when \( \alpha = 4 \), as demonstrated in Corollary 1.

\begin{corollary}
For $\alpha=4$, (\ref{eq:Pwu_exact}) can be rewritten as follows
\begin{equation}\label{eq:Pwu_exact_alpha4}
   P_{\text{wu}}^s = 1-2Q\left( \frac{\lambda\Delta}{\sqrt{2}} \sqrt{\frac{aP_t}{M_wT}} \right)
\end{equation}
 where \( Q(z) = \frac{1}{\sqrt{2\pi}} \int_z^\infty e^{-t^2/2} \, dt \), is the Gaussian CCDF.
\end{corollary}

It is worth noting that for the particular case, $N_t=M_w=1$, corresponding to the single-antenna scenario, \eqref{eq:Pwu_exact} and \eqref{eq:Pwu_exact_alpha4} match the results from \cite{kouzayha_joint_2018}, validating the single-antenna scenario in our analytical WUR-MIMO framework.

\subsection{False Wake-up Probability}
When the typical IoT device is not intended to wake up, the BS \( o \) updates its PA matrix, \( \textbf{P}_{o} \), by distributing power uniformly among the remaining \( M_w - 1 \) WSs and setting \( p_{1,o} = 0 \). A false alarm event occurs when the typical device is erroneously activated without any WUS transmitted from its serving BS. Thus, the false wake-up probability\footnote{The false wake-up probability is evaluated via Monte-Carlo simulations, and the analytical model remains a topic for future investigation.}, denoted as $P_{\text{wu}}^f$, is defined as the probability that the interference exceeds the power threshold $T$ and can be expressed as  
\begin{equation}\label{eq:falseWU}
 P_{\text{wu}}^f \triangleq \mathbb{P}(I\ge T).
\end{equation}
\section{Numerical Results and Discussion} \label{sec:simu} 
In this section, we validate the derived analytical results for the success wake-up probability by conducting Monte-Carlo simulations of the WUR-MIMO scheme. We further compare it with the single-antenna scenario, highlighting the benefits of integrating WUR with MIMO. Unless otherwise stated, the system parameters are $\lambda = 10 \times (0.5^2 \times \pi)^{-1}$ [BSs/Km$^2$], $P_t = 46$ [dBm], $\alpha=4$, and $a = 1$.

Fig.~\ref{fig:prob_SISO_anal} illustrates the effect of increasing the pathloss exponent $\alpha$ on the wake-up success probability for the case where $N_t=M_w=1$. The results validate the accuracy of the analytical expression derived in Theorem~\ref{theorem:Ps}, as there is a close match between the analytical results and those obtained from Monte-Carlo simulations across various threshold $T$ and pathloss exponent $\alpha$ values. Additionally, Fig. \ref{fig:prob_SISO_anal} demonstrates that the derived success probability expression for 
$N_t=M_w=1$ is consistent with the results reported in~\cite{kouzayha_joint_2018}, which correspond to the single-antenna WUR scenario. This alignment underscores the robustness of the proposed analytical model in capturing the behavior of single-antenna configurations.
\begin{figure}[t!]
    \centering \includegraphics[width=1\linewidth]{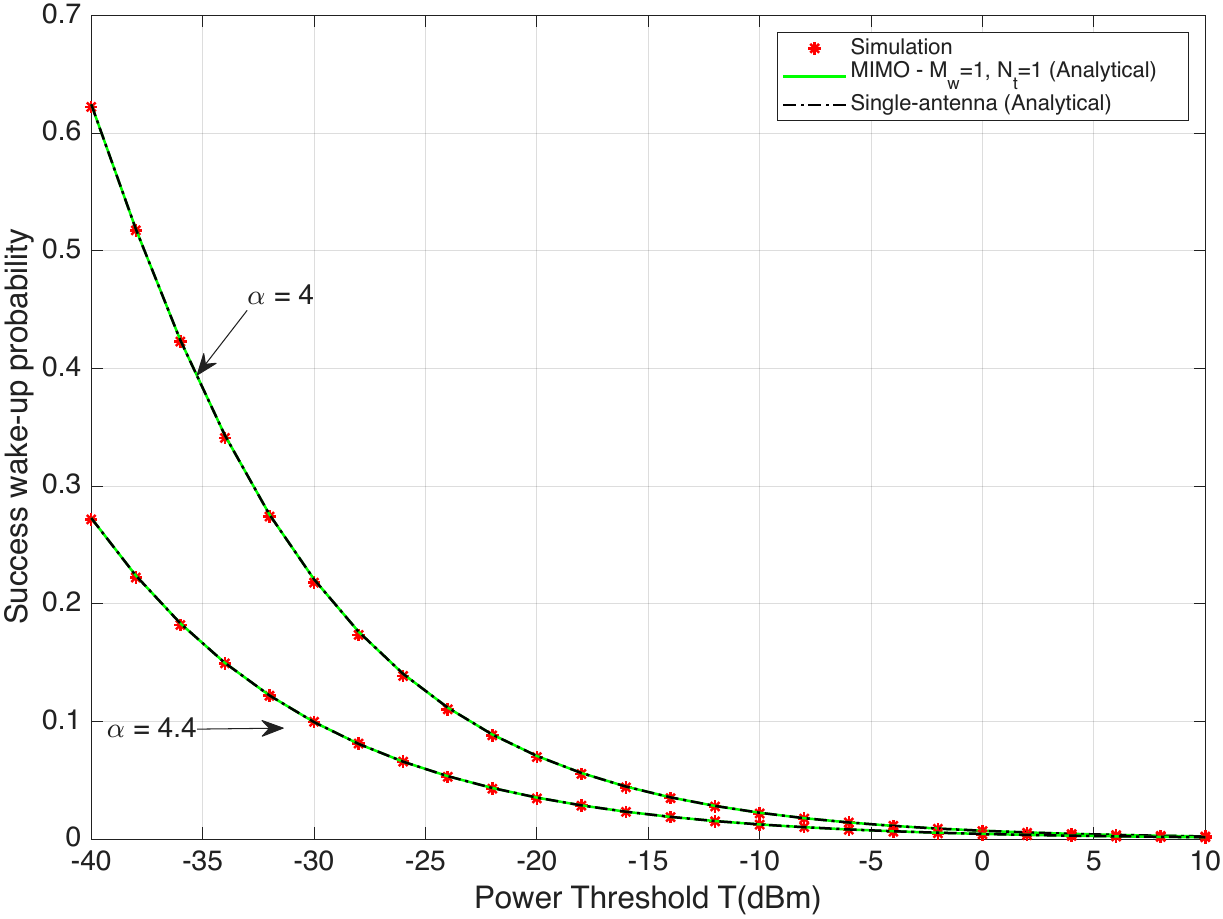}
    \caption{\small Success wake-up probability vs. the power threshold, $T$, for various pathloss exponents, $\alpha$, and single-antenna WUR. }  
    \label{fig:prob_SISO_anal}
\end{figure}

\begin{figure}[!t]
    \centering \includegraphics[width=0.95\linewidth]{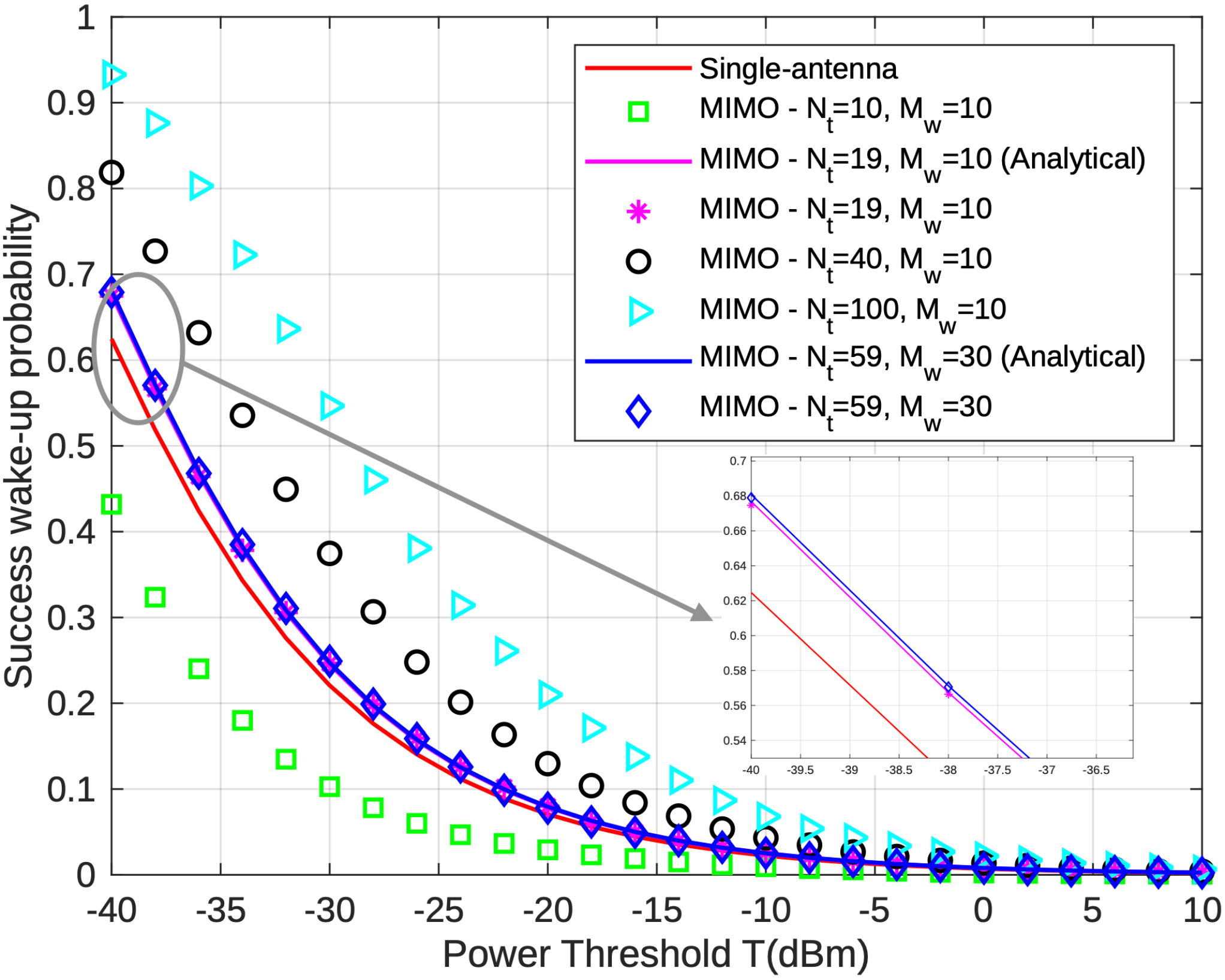}
    \caption{\small Success wake-up probability vs. the power threshold, $T$, for WUR-MIMO solutions. } 
    \label{fig:success_MIMO}
\end{figure}

Fig.~\ref{fig:success_MIMO} depicts the success wake-up probability, \(P_{\text{wu}}^s\) for varying numbers of BS antennas, \(N_t = \{1, 10, 19, 40, 100\}\), serving \(M_w = 10\) IoT devices per cell. Analytical results are provided only for $N_t=19$ and $M_w=10$ consistent with the assumption $N_t = 2M_w - 1$, while the remaining outcomes are simulation-based. The close agreement between analytical and simulation results, as demonstrated in Fig. \ref{fig:prob_SISO_anal}, supports the validity of the derived expression for success wake-up probability. Fig.~\ref{fig:success_MIMO} shows that \(P_{\text{wu}}^s\) generally increases with $N_t$, as more antennas at the BS allow for more focused energy transmission to the desired device, thereby enhancing the success wake-up probability. This improvement is attributed to the enhanced beamforming capabilities of multiple antennas, which concentrate energy more effectively on the desired device, boosting the wake-up success probability. For instance, WUR-MIMO outperforms the single-antenna case by $8\%$, $31\%$, and $49\%$  for \(N_t = 19, 40\) and \(100\) antennas, respectively. However, Fig.~\ref{fig:success_MIMO} also reveals performance degradation under certain conditions. Specifically, when  \(N_t = M_w = 10\), the success probability drops below that of the single-antenna case. This decline occurs because the number of antennas is insufficient to provide the spatial diversity needed for effective separation of IoT devices, underscoring the importance of adequate antenna configurations for optimal wake-up performance. In this context, Fig.~\ref{fig:success_MIMO} also depicts the success wake-up probability, \(P_{\text{wu}}^s\), with $M_w=30$ and $N_t=2M_w-1=59$, evaluated both analytically and through Monte-Carlo simulations. We observe that when the BS uses $N_t=2M_w-1$, \(P_{\text{wu}}^s\) in (\ref{eq:Pwu_exact}), which depends on $M_w$, remains approximately consistent for both scenarios, $M_w=10$ and $M_w=30$. Additionally, MIMO with $N_t=2M_w-1$ achieves a higher success wake-up probability than the single-antenna case, regardless of $M_w$. This behavior highlights the need to configure the number of antennas to exceed $2M_w-1$ for reliable IoT wake-up performance and to maintain the benefits of WUR-MIMO. In other words, a minimum \textit{antenna-to-device ratio} is required to maintain the MIMO advantage over the single-antenna baseline. To improve the trade-off between spectral and energy efficiency at the transmitter, hybrid beamforming with fewer RF chains than antennas is a promising direction. Its combination with low-complexity, low-feedback methods such as those in \cite{khaled2023angle} remains an interesting avenue for future work in integrated wake-up and MIMO systems.


\begin{figure}[!t]
    \centering \includegraphics[width=0.95\linewidth]{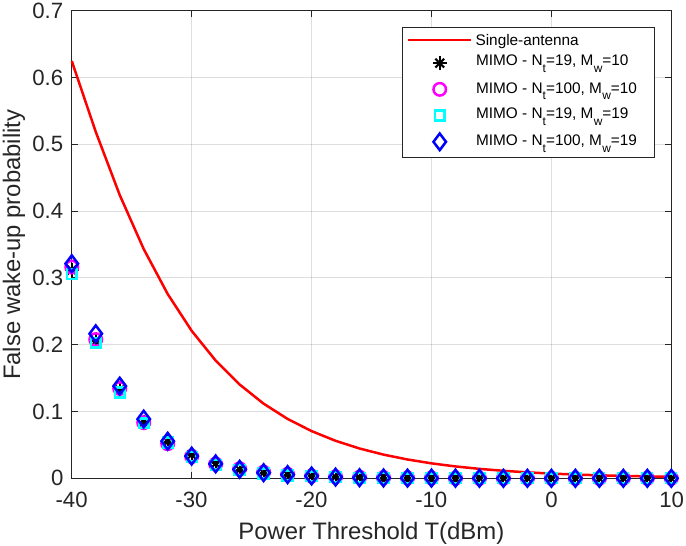}
    \caption{\small False wake-up probability vs. the power threshold, $T$, for single-antenna WUR and WUR-MIMO solutions} 
    \label{fig:false_MIMO}
\end{figure}

Fig.~\ref{fig:false_MIMO} plots the simulation results of the false wake-up probability, $P_{\text{wu}}^f$, in (\ref{eq:falseWU}) as a function of the power threshold for different values of $N_t$ and $M_w$. The results reveal a significant improvement brought by integrated wake-up and MIMO techniques in reducing the false wake-up probability. Compared with the single-antenna case, the false wake-up probability is reduced by $50\%$ (resp. $89\%$) when $T = -40$~dBm (resp. $T = -30$~dBm). The significant reduction in probability of false wake-up may translate into longer lifetime of sensors and reduction of spurious traffic on the wireless medium. Furthermore, the false wake-up probability remains nearly constant regardless of $N_t$ and $M_w$. Although the statistical characteristics of both intra-cell and inter-cell interference follow a gamma distribution with parameters dependent only on $M_w$ under ZF precoding, it was expected that the false probability would depend solely on $M_w$. Similar to the behavior observed in the expression for  $P_{\text{wu}}^s$ in (\ref{eq:Pwu_exact}) with $N_t=2M_w-1$, $P_{\text{wu}}^f$ remains approximately consistent regardless of $M_w$.  This phenomenon is attributed to the functional forms of $P_{\text{wu}}^s$ and $P_{\text{wu}}^f$ which will be further analyzed in future work.

\section{Conclusion}\label{sec:conclu}
In this paper, we investigate the integration of WUR and MIMO in cellular IoT networks. We develop a tractable analytical framework to assess the performance of the WUR-MIMO scheme under the antenna configuration $N_t = 2M_w - 1$, where $M_w$ denotes the number of devices. The analytical results, validated through extensive Monte-Carlo simulations, demonstrate that MIMO beamforming enhances WUR performance, once this antenna configuration is employed. The findings provide practical insights on the minimum antenna-to-device ratio for the deployment and configuration of WUR-MIMO systems. Moreover, the results show that, compared with the single-antenna baseline, MIMO significantly reduces the probability of false wake-ups while increasing the probability of successfully activating the intended devices. Overall, this work advances the understanding of WUR-MIMO as a solution to improve device energy saving and extend the WUS coverage. With WUR now standardized in 3GPP Release 18 (Technical Report 38.869) for wireless networks, a promising future direction is to extend WUR-MIMO to integrated sensing and communication networks, enabling greener and more sustainable wireless systems.

\section*{Acknowledgment}
This work was supported by the Council of the Region Bretagne, under the grant ENERCAP, and was carried out during a postdoctoral fellowship at IMT Atlantique.

\appendix
By performing the inverse Laplace transform of $\mathcal{L}_{{P_r}}(s)$ and defining $ F(s) \triangleq \int_{\gamma-i\infty}^{\gamma+i\infty}\frac{1}{s}e^{-\lambda \left[\frac{aP_t}{M_w}s\right]^{\frac{2}{\alpha}} \Delta}e^{st}ds$, the CCDF of $P_r$, denoted as $P_{\text{wu}}^s$, can be calculated as follows

\begin{equation} \label{eq:P_wu_1}
P_{\text{wu}}^s = 1-\frac{1}{2\pi i}\int_{c-i\infty}^{c+i\infty} F(s) ds,
\end{equation}
where $c$ is a vertical contour in the complex plane chosen to the left of all singularities of $F(s)$. To compute $\int_{c-i\infty}^{c+i\infty} F(s) ds$, we use the complex inversion integral formula for Laplace transforms with $s=0$ as a branch point and the appropriate Bromwich contour avoiding it. Thus, $I_C$ is given by
\begin{equation}
  I_C=  \oint\limits_{C}F(s)\,\mathrm{d}s=\int\limits_{C_1}+\int\limits_{C_2}+\int\limits_{C_3}+\int\limits_{C_4}+\int\limits_{C_5}+\int\limits_{C_6}F(s)ds=0,
\end{equation}
Refer to Fig. \ref{fig:BC} for contour details. $C_1$ is the vertical line segment from $c-iR$ to $c+iR$.  $C_2$ is the circular arc of radius $R$ from the top of $C_1$ to just above the negative real axis. $C_3$ is the horizontal line segment just above the negative real axis from $[-R, -\epsilon]$, with $s=xe^{i\pi}$. $C_4$ is the circular arc of radius $\epsilon$ around the origin, with  $s=\epsilon e^{i\phi}$.  $C_5$  is horizontal line segment just below the negative real axis from  $[-\epsilon, -R]$, with $s=xe^{-i\pi}$.  $C_6$ the circular arc of radius $R$ from just below the negative real axis to the bottom of $C_1$. For $t>0$, the integrals over $C_2$ and $C_6$ vanish as $R \rightarrow \infty$.  
So $I_C$ is given by 
\begin{multline}
    I_C= \int_{c-i\infty}^{c+i\infty}F(s)ds+\int_{\infty}^{\epsilon}e^{-\lambda(\frac{aP_t}{M_w}e^{i\pi}x)^{2/\alpha}\Delta}\frac{e^{-xt}}{x}dx \\ +\int\limits_{\pi}^{-\pi}ie^{-\lambda(\frac{aP_t}{M_w}e^{i\epsilon\phi}x)^{2/\alpha}\Delta(-\epsilon e^{i\phi})}e^{\epsilon e^{i\phi}}d\phi \\+\int_{\epsilon}^{\infty}e^{-\lambda(\frac{aP_t}{M_w}e^{-i\pi}x)^{2/\alpha}\Delta}\frac{e^{-xt}}{x}dx=0,
\end{multline}
Note the apparent singularity at $\epsilon=0$; however, the divergence cancels as $\epsilon \longrightarrow 0$. In this limit, the third integral can be given by $-2i\pi$. To simplify further, we can re-scale and combine the second and fourth integrals. By letting  $e^{i2\pi/\alpha}=\cos(2\pi/\alpha)+i\sin(2\pi/\alpha)$  and  $e^{ix}-e^{-ix}=2i\sin(x)$ , we obtain 
\begin{multline}
    \label{eq:I_cc}
    \resizebox{0.455\textwidth}{!}{$
    \int_{c-i\infty}^{c+i\infty} F(s)\, ds = 2i\pi - \int_{0}^{\infty} \frac{e^{-xt}}{x} e^{-\lambda \left( \frac{aP_t}{M_w} \right)^{2/\alpha} \cos\left(\frac{2\pi}{\alpha}\right) x^{2/\alpha} \Delta} $} \\
    \resizebox{0.35\textwidth}{!}{$
    - 2i \sin\left( \lambda \left( \frac{aP_t}{M_w} \right)^{2/\alpha} \sin\left( \frac{2\pi}{\alpha} \right) x^{2/\alpha} \Delta \right) dx.$}
\end{multline}

By replacing (\ref{eq:I_cc}) in (\ref{eq:P_wu_1}), Theorem~\ref{theorem:Ps} is verified.
\begin{figure}[!t]
\vspace{-0.5cm}
    \centering \includegraphics[scale=0.15]{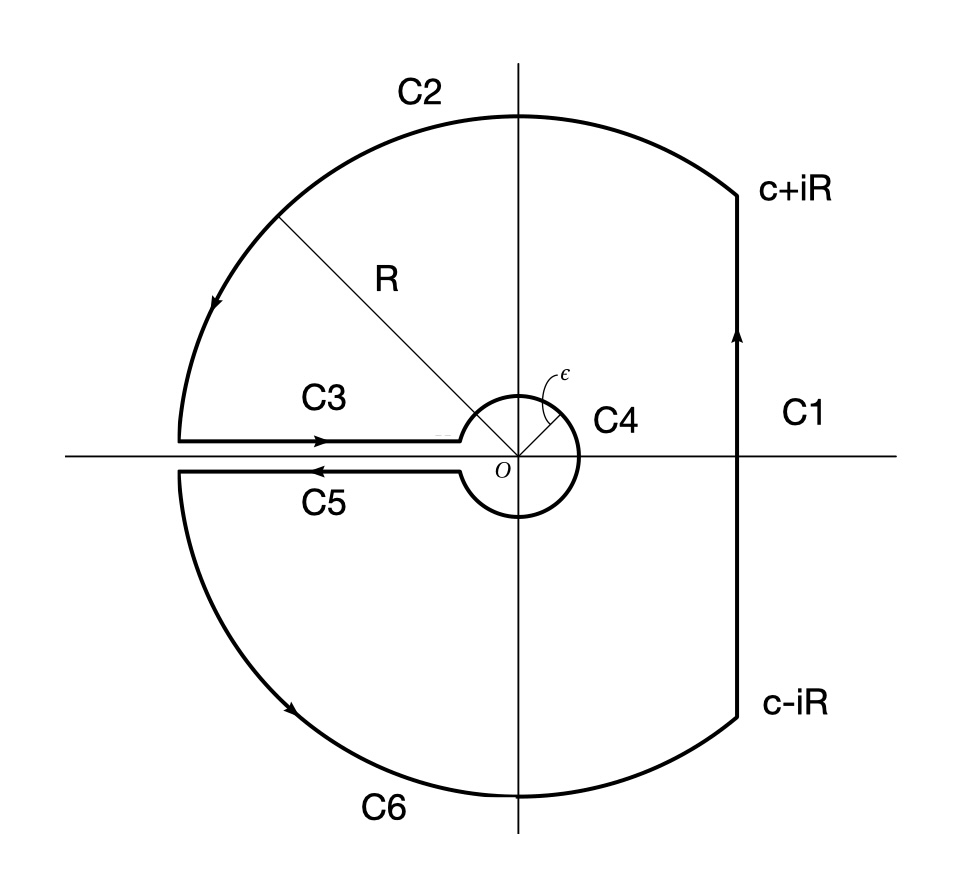}
    \caption{\small Bromwich contour that avoids the branch point $s=0$. }    
    \label{fig:BC}
\end{figure}

\bibliographystyle{IEEEtran}
\bibliography{Bibliography}


\end{document}